\documentclass[11pt, letter]{article}
\usepackage[margin=2cm]{geometry}

\usepackage{aligned-overset}
\usepackage{microtype}
\usepackage{graphicx}
\usepackage{subfigure}
\usepackage{booktabs} 
\usepackage{amsmath}
\usepackage{amsthm}
\usepackage{amssymb}
\usepackage{wrapfig}
\usepackage{graphicx}
\usepackage[dvipsnames]{xcolor}  
\usepackage{multirow}
\usepackage{epsfig,amssymb,amsfonts,amsmath,amsthm} 
\usepackage{appendix}
\usepackage[ruled,linesnumbered]{algorithm2e}
\usepackage{hhline}  
\usepackage{times}
\usepackage{placeins} 

\usepackage{xcolor}
\usepackage[shortlabels]{enumitem}
\usepackage[T1]{fontenc}    
\usepackage[none]{hyphenat}

\definecolor{newblue}{rgb}{0.2,0.2,0.6} 

\usepackage[pagebackref=false,colorlinks=false,bookmarks=false]{hyperref}

\newcommand{\uhat}{\widehat u}

\newcommand{\trace}{ \mathrm{tr}}
\newcommand{\seta}{\mathcal{A}}
\newcommand{\setb}{\mathcal{B}}

\newtheorem{theorem}{Theorem}

\newtheorem{lemma}{Lemma}[section]

\newtheorem{remark}{Remark}

\DeclareMathOperator*{\argmin}{arg\,min}        
\DeclareMathOperator*{\argmax}{arg\,max}        

\newcommand{\norm}[1]{\left\| #1\right\|}                  
\newcommand{\identity}{\mathbb{I}}

\newcommand{\I}{\identity}

\newcommand{\R}{\mathbb{R}}
\newcommand{\C}{\mathbb{C}}

\newcommand{\abs}[1]{\left\lvert#1\right\rvert}
\newcommand{\cardinality}[1]{\abs{#1}}

\newcommand{\poly}[1]{\mathrm{poly}\!\left(#1\right)}

\definecolor{indiagreen}{rgb}{0.07, 0.53, 0.03}

\title{Polynomial-Time Algorithms for Weaver's Discrepancy Problem \\in a Dense Regime}

\author{
Ben Jourdan\footnote{School of Informatics,    University of Edinburgh, UK. \url{ben.jourdan@ed.ac.uk}}
\and 
Peter Macgregor\footnote{School of Informatics,    University of Edinburgh, UK. \url{peter.macgregor@ed.ac.uk}. }
\and He Sun\footnote{School of Informatics, University of Edinburgh, UK. \url{h.sun@ed.ac.uk}. }  }

\begin{document}

\maketitle

\thispagestyle{empty}
\setcounter{page}{0}

\begin{abstract}

Given $v_1,\ldots, v_m\in\mathbb{C}^d$ with $\|v_i\|^2= \alpha$ for all $i\in[m]$ as input and 
suppose
$
\sum_{i=1}^m | \langle u, v_i \rangle |^2 = 1
$
for every unit vector  $u\in\mathbb{C}^d$, Weaver's discrepancy problem asks for a partition $S_1, S_2$ of $[m]$, such that $
\sum_{i\in S_{j}} |\langle u, v_i \rangle|^2 \leq 1 -\theta$
for some universal constant $\theta$,
every unit vector $u\in\mathbb{C}^d$ and every $j\in\{1,2\}$. We prove that this problem can be solved   
deterministically in   polynomial time when $m\geq  49 d^2$. 

\end{abstract}

\newpage

\section{Introduction}

The   Kadison-Singer problem is a central problem in operator theory, and has close connections to a number of problems in quantum mechanics, pure and applied mathematics, 
engineering, and computer science.  Among several formulations of the Kadison-Singer problem,  Weaver~\cite{dm/Weaver04} shows
that it is
equivalent to the following discrepancy problem~(the $\mathsf{KS}_2$ problem): 
there  exists a universal constant  $\theta>0$ such that the following holds. Let $v_1,\ldots, v_m\in\mathbb{C}^d$ satisfy $\|v_i\|^2= \alpha$ for all $i\in[m]$, and suppose
$
\sum_{i=1}^m | \langle u, v_i \rangle |^2 = 1
$
for every unit vector $u\in\mathbb{C}^d$. Then, there exists a partition $S_1, S_2$ of $[m]$, such that $
\sum_{i\in S_{j}} |\langle u, v_i \rangle|^2 \leq 1 -\theta$
for every unit vector $u\in\mathbb{C}^d$ and every $j\in\{1,2\}$.

In their breakthrough result, Marcus, Spielman, and Srivastava~\cite{MSS} present a non-constructive proof and show that the partition promised by the  $\mathsf{KS}_2$ problem exists, leading  to an affirmative answer to the Kadison-Singer problem. On the other side, due  to a close connection between the $\mathsf{KS}_2$ problem and many algorithmic problems~(e.g.,  constructing unweighted spectral
sparsifiers,  and spectrally thin trees), 
they ask 
whether such a  partition of the $\mathsf{KS}_2$ problem can be found in polynomial time.
After a decade-long research~(e.g., \cite{conf/soda/AnariGSS18,JMS23,SZ22}), this problem remains wide
 open and has become a very important open problem in algorithmic spectral graph theory.

This paper studies  the   $\mathsf{KS}_2$ problem in the regime of $m\geq 221d^2$. Under this condition, we present two deterministic polynomial-time algorithms that solve the  $\mathsf{KS}_2$ problem. These two algorithms achieve the same approximation guarantee, and their performance is 
summarised as follows:

\begin{theorem}
\label{theorem:alg1}
Let $\mathcal{I}=\{v_i\}_{i=1}^m $ be vectors in  $\mathbb{C}^d$, such that $m$ is even and 
$
\sum_{i=1}^m | \langle u, v_i \rangle |^2 = 1
$
for every unit vector $u \in \mathbb{C}^d$. Moreover, assume that
$\| v_i\|^2=\alpha$ for every $i\in[m] $ and $m\geq 221 d^2$. Then, there is a deterministic  algorithm that finds a partition $S_1, S_2$ of $[m]$
such that  $|S_1|=|S_2|$ and 
\[
\sum_{i\in S_j} |\langle u, v_i\rangle|^2 \leq \frac{3}{4}
\]
 for every unit vector $u \in \mathbb{C}^d$ and every $j\in\{1,2\}$.
The algorithm runs in time $O(\poly{m, d})$.
\end{theorem}

\begin{remark}
Our second algorithm   requires $m\geq 49d^2$ instead. Both constants might be further optimised.
\end{remark}

\begin{remark}
Our presented algorithms can be easily adjusted to work when $m$ is odd. In this case,  $|S_1|\neq|S_2|$ but the other stated   properties from Theorem~\ref{theorem:alg1} hold.
\end{remark}

We discuss the significance of the result. First of all, our result shows that, when $m\geq 49 d^2$,  the $\mathsf{KS}_2$ problem can be solved by a 
deterministic  and  polynomial-time algorithm.
To the best of our knowledge, this  presents the first such algorithm which does not rely on random sampling and de-randomization techniques~\cite{wigderson2008derandomizing}.
Secondly, our work draws a novel connection between the $\mathsf{KS}_2$ problem and the determinant optimisation problem, and  demonstrates  how potential functions can be used for   the $\mathsf{KS}_2$ problem.  
Potential functions of various forms are common tools used in constructing spectral sparsifiers~\cite{BSS,LS17,LS18}; however, all of the previous works 
(i) might choose the same vector multiple times, and 
(ii) need to reweight the chosen vector.
Hence, the previous analyses cannot be directly applied for the 
$\mathsf{KS}_2$ problem. In our point of view, overcoming these two ``bottlenecks'' 
is significant, and our presented technique could motivate more research on this important problem.

\subsection{Overview of our Techniques}

\paragraph{The First Algorithm.}
At a very high level, our algorithm proceeds in $m/2$ iterations, and picks one vector in each iteration. Formally, starting with
$A_0= \mathbf{0}_{d\times d}$ and $\seta_0=\emptyset$,
the algorithm picks some vector $v\in \mathcal{I} \setminus \seta_{j}$ in iteration $j$, and adds it into $\seta_{j+1}$, i.e., $\seta_{j+ 1}
\triangleq \seta_{j} \bigcup \{v\}$; the algorithm also sets $A_{j+1} = A_j + vv^*$.
Despite the similarity, it is important to notice the difference between our framework and the   BSS one~\cite{BSS}: first of all,  in each iteration our algorithm only picks a vector that hasn't been chosen yet, while the BSS algorithm might pick the same vector multiple times.  Secondly,
our algorithm doesn't scale the chosen vector, while the BSS algorithm scales the chosen vector in each iteration. Hence, during the execution of our algorithm, the set of the chosen vectors and the set of the remaining ones always form a partition of $\mathcal{I}$.

The analysis of our algorithm is based on a novel potential function  defined by
\[
\Phi^{u} (A) \triangleq \trace\log (u\I - A)^{-1}
\]
for some positive semi-definite matrix $A$ and barrier value $u\in\mathbb{R}^+$. Since $\trace\log (A) = \log\det (A)$ holds for any positive definite matrix $A$, one can rewrite $\Phi^u (A)$ as 
\begin{equation}\label{eq:potential_rewrite}
\Phi^{u} (A)  = - \log \det (u\I - A) = -\log \left(\prod_{i=1}^d \lambda_i(u\I-A)\right),
\end{equation}
i.e., the potential function $\Phi^u(A)$ is a function of the determinant of the matrix $u\I-A$. Our objective is to apply this potential function to keep track of the algorithm's progress in each iteration. To achieve this, we set the  initial barrier value as $u_0\triangleq 1/2$, which increases by $\delta_u$ after each iteration. That is, $u_{j+1} \triangleq u_j + \delta_u$, and we set $\delta_u\triangleq \alpha/d$.

Next, we reason about the use of $\Phi^u(A)$ and our choices of the barrier values.  As we set $\delta_u = \alpha/d$, it holds that
 \begin{equation}\label{eq:invariant}
 \trace(u_{j}\I - A_{j}) = \left(u_0 + \frac{\alpha}{d}\cdot j\right)\cdot d - \alpha\cdot j =\trace(u_0\cdot\I),
 \end{equation}
hence the trace of  $u_j\I - A_j$ remains  constant during the execution of the algorithm. On the other side, by the AM-GM inequality we have  that
\begin{equation}\label{eq:upperbound}
\det\left(u_j\I- A_j\right)^{1/d} \leq \frac{1}{d}\cdot\trace\left( u_j\I - A_j\right).
\end{equation}
Combining \eqref{eq:invariant} with \eqref{eq:upperbound}, we know that the potential function $\Phi^{u_j}(A_j)$ has the same upper bound for all the iterations; moreover,  the closer $\Phi^{u_j}(A_j)$ is to this upper bound,  the more  balanced are the eigenvalues of $u_j\I- A_j$. Hence, in every iteration $j$ the algorithm picks the vector $v\in\mathcal{I}\setminus \seta_j$ that maximises $\det\left(u_{j+1}I - A_j- vv^*\right)$.
We show that after $m/2$ iterations the left and right sides of \eqref{eq:upperbound} are sufficiently close. Due to a tight bound of the condition number of any Hermitian positive definite matrix $B \in \C^{d \times d}$ based on $\det(B)$, $\trace(B)$ and $d$~\cite{condition}, we show that the algorithm finds the partition promised by the $\mathsf{KS}_2$ problem.

\paragraph{The Second Algorithm.} Our second algorithm proceeds  in $m/2$ iterations as well, and picks one vector in each iteration. Formally,  starting with $A_0=\mathbf{0}_{d\times d}$ and $\seta_0=\emptyset$, the algorithm picks some vector $v\in \mathcal{I} \setminus \seta_{j}$ in iteration $j$, and adds it into $\seta_{j+1}$; the algorithm  sets $A_{j+1} = A_j +vv^*$. However, in each iteration $j$ the algorithm picks $v$ that minimises $v^*(\I - A_j)^{-1}v$. Instead of applying a potential function, we analyse the second algorithm by directly lower bounding $\det(\I - A_{m/2})$. We prove that such a lower bound of $\det(\I - A_{m/2})$ is sufficient to bound the condition number of $A_{m/2}$ and $\I - A_{m/2}$, implying that our algorithm's output is the partition promised by the  $\mathsf{KS}_2$ problem.

\subsection{Related Work}

Our result directly relates to the quest for efficient algorithms for the the  $\mathsf{KS}_2$ problem.  Anari et al.~\cite{conf/soda/AnariGSS18} show that a valid
partition promised by the  $\mathsf{KS}_2$ problem can be found in
$d^{O(m^{1/3}\alpha^{-1/4})}$ time.
Jourdan et al.~\cite{JMS23} present a randomised algorithm for the problem, which runs in time quasi-polynomial in $m$ and exponential in $d$. Compared with their results, we show that the $\mathsf{KS}_2$ problem  can be solved in polynomial time when $m\geq 49 d^2$. Weaver~\cite{weaver2} gives a polynomial-time algorithm for a one-sided version of the $\mathsf{KS}_2$ problem based on the BSS framework for spectral sparsification~\cite{BSS}.
In comparison, we give a two-sided guarantee in a more restricted setting.

Our work is also linked to hardness results for the $\mathsf{KS}_2$ problem~\cite{SZ22, JMS23}.
These results show that the optimisation version of the $\mathsf{KS}_2$ problem is $\mathsf{NP}$-hard.
It's worth mentioning  that the two reductions 
shown in \cite{SZ22,JMS23} are based on instances satisfying $m=\Theta(d)$.

Our approach further relates to the determinant maximisation problem which has been extensively studied in the field of experimental design~\cite{brown2022determinant, brown2024maximizing, SX-determinant}.
Compared with the state-of-the-art determinant maximisation algorithm, our approach can be applied to a wider range of instances.
We further discuss this connection in Section~\ref{sec:second}.

\section{Preliminaries\label{sec:pre}}

Let $\mathcal{I} \triangleq \{v_1,\dots, v_m \}$ be the set of input vectors, where $\| v_i\|^2=\alpha$
and 
$v_i\in \mathbb{C}^d$ for every $1\leq i \leq m$. We  use $\seta_j$ to represent the set of vectors picked till the $j$th iteration, and let $\setb_j \triangleq \mathcal{I} \setminus \seta_j$ be the set of unpicked vectors; hence, it holds that $|\seta_j| = j$ and $|\setb_j| = m-j$. We always use $A_j$ to represent the matrix constructed in the $j$th iteration, i.e., $A_j \triangleq \sum_{v\in \seta_j} vv^{*}$. We write the eigenvalues of any Hermitian matrix $A\in\mathbb{C}^{d\times d}$ as $\lambda_{\max}(A)=\lambda_1(A)\geq\ldots\geq \lambda_d(A)=\lambda_{\min}(A)$, and the condition number of $A$ as 
\[
\kappa(A)\triangleq \frac{\lambda_1(A)}{\lambda_d(A) }.
\]
The following relationship between $\alpha,d$ and $m$ will be repeatedly used in our analyses.

\begin{lemma}\label{lem:amd}
It holds that $m\cdot \alpha=d$.
\end{lemma}
\begin{proof}
Let $e_1,\ldots, e_d$ be the vectors forming the standard orthonormal basis, and we have that
\[
m\cdot\alpha = \sum_{i=1}^m \|v_i\|^2 = 
\sum_{i=1}^m \sum_{j=1}^d \langle v_i, e_j\rangle^2 =  \sum_{j=1}^d \sum_{i=1}^m\langle v_i, e_j\rangle^2 =d, 
\]
which proves the statement.
\end{proof}

Next we list the facts on matrices that will be used in our analyses.

\begin{lemma}[Theorem~3.30, \cite{alma9924127819402466}]\label{lem:det_product}
It holds for any square matrices $A, B$ of the same size that  $\det(AB) = \det(A)\cdot \det(B)$.
\end{lemma}

\begin{lemma} \label{lem:det_vs_tr}
It holds for any Hermitian positive definite matrix $A\in\mathbb{C}^{d\times d}$ that
\[
\det(A)^{1/d} \leq \frac{1}{d}\cdot\trace(A).
\]
\end{lemma}
\begin{proof}
Since we can write $\det(A)=\prod_{i=1}^d \lambda_i(A)$ and $\trace(A) = \sum_{i=1}^d \lambda_i(A)$, the statement holds by the AM-GM inequality.
\end{proof}

\begin{lemma}[Matrix Determinant Lemma, \cite{DING20071223}]\label{lem:MDL}
If $A$ is an invertible matrix and $u,v$ are vectors, then 
\[
\det \left(A+ uv^{*}\right) = \left(1 + v^{*} A^{-1} u\right)\cdot \det(A).
\]
\end{lemma}

\begin{lemma}
    \label{lem:matrixlogrules}
    Let $A, B$ be 
   positive definite matrices. Then, the following statements hold: 
    \begin{enumerate}
        \item $\trace\log(A)=\log\det(A)$;
        \item $\trace \log(AB) = \trace\log(A) + \trace \log(B)$;
        \item if $A$ and $B$ commute, then $\log(AB)=\log(A)+\log(B)$;
        \item $\log(A^{-1})=-\log(A)$.
    \end{enumerate}
\end{lemma}
\begin{proof}
Let $\lambda_1\ge \dots \ge \lambda_d>0$ be the eigenvalues of $A$. Then, we have that 
    \[
    \trace \log(A) = \sum_{i=1}^d\log(\lambda_i) = \log\left(\prod_{i=1}^d \lambda_i\right) = \log(\det(A)),
    \]
    which proves the first statement. 
    By the first statement  and Lemma~\ref{lem:det_product}, we have that 
    \[
    \trace \log(AB) = \log(\det(AB))= \log\det(A) + \log\det(B) = \trace \log(A)+\trace\log(B),
    \]
    hence the second statement holds. The proofs of the third and fourth statements can be found in Theorems 11.2 and 11.3 of Chapter 11 of \cite{HighamBook}.
\end{proof}

\begin{lemma}[\cite{condition}] \label{lem:condition}
Let $A\in\mathbb{C}^{d\times d}$ be a Hermitian  positive definite matrix. Then, the condition number $\kappa(A)$ of $A$ satisfies that 
\[
\kappa(A) \leq \frac{1+x}{1-x}, \qquad x \triangleq \sqrt{1 - \left( \frac{d}{\trace(A)} \right)^d\cdot\det(A) }. 
\]
\end{lemma}

\section{The First Algorithm and Analysis\label{sec:first}}

This section 
presents and analyses our first algorithm, and is organised as follows. Section~\ref{sec:algo1} gives the formal description of our first algorithm, whose correctness is analysed in Section~\ref{sec:proof1}; Section~\ref{sec:keylemma} proves  a key lemma used in our analysis.

\subsection{Algorithm Description\label{sec:algo1}}

Our designed  algorithm proceeds in $m/2$ iterations, and picks one vector in each iteration. At the initialisation step, the algorithm sets
\[
A_0 = \mathbf{0}_{d\times d}, \quad \seta_0=\emptyset, \quad \setb_0=\mathcal{I},
\]
and  $u_0=1/2,\delta_u=\alpha/d$. Then, in iteration $0\leq j< m/2$,  the algorithm selects a vector $v\in \setb_j$  and performs the following operations:
\[
A_{j+1} = A_{j} + vv^*, \qquad \seta_{j+1} = \seta_j \cup \{v\} 
, \qquad \setb_{j+1} = \setb_j \setminus \{v\}.
\]
The vector $v$ is chosen such that the increase of the potential function from
$\Phi^{u_{j}}(A_j)$ to $\Phi^{u_{j}+\delta_u}(A_j+vv^*)$
is minimised.  That is, the algorithm picks $v$ that  minimises the function
\[
    \Phi^{u_j +\delta_u}(A_j + v v^*) = -\log\Big( \det((u_j + \delta_u)\cdot \I-A_j-vv^*)\Big).
\]
After this, the algorithm sets $u_{j+1} = u_j +\delta_u$, and   moves to the next iteration.
See  Algorithm~\ref{alg:alg1} for formal description.

\begin{algorithm}[th!]
\SetAlgoCaptionSeparator{}
\caption{\label{alg:alg1}}
\textbf{Input}: $\mathcal{I} = \{v_i\}_{i = 1}^m$, where $v_i \in \C^d$ and $\norm{v_i}^2 = \alpha$\\
$A_0 \gets \mathbf{0}_{d \times d}$, $\seta_0\gets \emptyset$\\
$\setb_0 \gets \mathcal{I}$\\

$u_0 \gets 1/2$ \\
$\delta_u \gets \alpha / d$ \\
$j \gets 0$  \\ 
\While{$j< m/2$}{
    $u_{j+1} \gets u_j + \delta_u$  \\
    $v_j \gets \argmax_{v \in \setb_j} \det(u_{j+1} \identity - A_j - v v^*)$ \label{alg:alg1_potential_minimisation} \\
    $A_{j+1} \gets A_j + v_j v_j^*$ \\
    $\seta_{j+1} \gets \seta_j \cup \{v_j\}$\\
    $\setb_{j+1} \gets \setb_j \setminus \{v_j\}$\\
    $j \gets j+1$
}
\Return $\seta_{m/2}$
\end{algorithm}

\subsection{Proof of Theorem~\ref{alg:alg1} \label{sec:proof1}}

We introduce $c_j\triangleq u_j-\lambda_{\max}(A_j)$, which is  the difference between $u_j$ and the maximum eigenvalue of $A_j$ constructed by the algorithm in  iteration $j$. Recall that Theorem~\ref{alg:alg1} assumes that $m\geq 221d^2$, which is equivalent to 
\[
\alpha\leq \frac{1}{221d};
\]
we assume that this holds throughout the rest of the section. We first claim that the constructed $A_j$ in every iteration satisfies the following:
\begin{enumerate}
    \item the maximum eigenvalue $\lambda_{\max}(A_j)$ is \emph{far away} from   $u_j$ in every iteration;
    \item the condition number $\kappa(u_j\I - A_j)$ is at most $3/2$ in every iteration.
\end{enumerate}
These facts are summarised in Lemma~\ref{lem:gap_bound}, which will be proven in the next subsection.

\begin{lemma} \label{lem:gap_bound}
    It holds for every
    $0 \leq j \leq m/2$
    that   $c_j\geq 1/3$
    and $\kappa(u_j\I -A_j) \le 3/2$.
\end{lemma}
Lemma~\ref{lem:gap_bound} is sufficient to prove Theorem \ref{theorem:alg1}.
\begin{proof}[Proof of Theorem \ref{theorem:alg1}]
After $m/2$ iterations of Algorithm~\ref{alg:alg1} we have \[
u_{m/2} = u_0 + \frac{m}{2}\cdot \delta_u= \frac{1}{2} + 
\frac{m}{2}\cdot\frac{\alpha}{d} = 1,
\]
using the fact $d=m\alpha$~(Lemma~\ref{lem:amd}). By Lemma~\ref{lem:gap_bound}, we have that $\lambda_{\mathrm{max}}(A_{m/2}) \leq 2/3$ and $\kappa(\I - A_{m/2}) \leq 3/2$. This  implies that
    \[
        \lambda_{\mathrm{max}}(\I - A_{m/2}) \leq \frac{3}{2}\cdot  \lambda_{\mathrm{min}}(\I - A_{m/2}),
    \]
    which is equivalent to
    \[
        1 - \lambda_{\mathrm{min}}(A_{m/2}) \leq \frac{3}{2}\cdot  (1 - \lambda_{\mathrm{max}}(A_{m/2})).
    \]
    Combining this with the fact that
    \[
        \lambda_{\mathrm{max}}(A_{m/2}) \geq \frac{1}{d} \cdot \trace(A_{m/2})= \frac{1}{d}\cdot \frac{\alpha m}{2} = \frac{1}{2}
    \]
    gives us that 
    $
        \lambda_{\mathrm{min}}(A_{m/2}) \geq 1/4$.
    The algorithm performs $O(m)$ iterations, each of which consists of elementary matrix operations, and $O(m)$ determinant computations, each of which has running time $O(d^3)$. 
    Thus, the running time of the algorithm is $O(\mathrm{poly}(m, d))$.
\end{proof}

\subsection{Proof of Lemma~\ref{lem:gap_bound}\label{sec:keylemma}}

This subsection analyses the average change of the potential function $\Phi^{u_j}(A_j)$ in each iteration, and 
proves by induction that Lemma~\ref{lem:gap_bound} holds for every iteration. We first assume that Lemma~\ref{lem:gap_bound} holds for any iteration $j'\leq j$ for some fixed $j$, and only study the $j$th iteration of Algorithm~\ref{alg:alg1}. For simplicity, we write  \[
A=A_j, \quad B=B_j,\quad \seta= \seta_j, \quad \setb=\setb_j,
\]
as well as $u = u_j$ and $\uhat = u_{j+1} = u_j + \delta_u$.

\begin{lemma}\label{lem:pot_drop_part1}
   Assuming  Lemma~\ref{lem:gap_bound} holds up to iteration $j$,  we have  for the constructed  matrix $A=A_j$  that 
   \begin{align}
\lefteqn{\sum_{v\in\setb}  \Phi^u(A) - \Phi^{\uhat} (A+vv^{*})}\nonumber\\
     & =(m-j)\cdot \trace \left[\log\Big((u\I - A)^{-1}(\uhat\I-A)\Big)\right] + \sum_{v\in\setb} \log\Big( 1-v^{*}(\uhat\I-A)^{-1}v\Big).\label{eq:lowerbd1}
    \end{align}
\end{lemma}
\begin{proof}
By Lemma~\ref{lem:gap_bound} we have that $\lambda_{\min}(u\I-A)=u-\lambda_{\max}(A)\ge 1/3>0$ and \[
\lambda_{\min}(u\I - A-vv^*) \ge \lambda_{\min}(u\I - A-\alpha \I)\ge \frac{1}{3}  - \alpha \ge \frac{1}{3} - \frac{1}{221d}>0;
\]
as such  both of  $(u\I-A)$ and $(\uhat \I - A - vv^*)$ are invertible.
Hence, it holds for any matrix $A$ and vector $v\in\mathbb{R}^d$ that
\begin{align*}
    \Phi^u(A) - \Phi^{\uhat}(A+vv^{*}) & = \trace \Big[\log\Big((u\I - A)^{-1}\Big)- \log\Big((\uhat \I - (A+vv^{*}))^{-1}\Big)\Big]\\
    &= \trace \Big[\log\Big((u\I - A)^{-1}\Big) + \log\Big(\uhat \I - (A+vv^{*})\Big)\Big]\\
    &= \trace \Big[\log\Big((u\I - A)^{-1}(\uhat \I - (A+vv^{*}))\Big) \Big],
\end{align*}
where we apply the definition of $\Phi^u(A)$ and  Lemma~\ref{lem:matrixlogrules}.
Thus,
\begin{align}
    \sum_{v\in\setb}  \Phi^u(A) - \Phi^{\uhat} (A+vv^{*}) 
    &= \sum_{v\in\setb} \trace \Big[\log\Big((u\I - A)^{-1}(\uhat \I - (A+vv^{*}))\Big) \Big] \nonumber\\
    &= \sum_{v\in\setb} \log \Big[ \det \Big((u\I - A)^{-1}(\uhat \I - (A+vv^{*}))\Big) \Big]\nonumber\\
      &=\log \left[ \prod_{v\in\setb} \det\Big((u\I - A)^{-1}(\uhat \I - (A+vv^{*}))\Big) \right] \nonumber\\
    &=\log \left[ \prod_{v\in\setb} \det\Big((u\I - A)^{-1}\Big)\det\Big(\uhat \I - (A+vv^{*})\Big) \right] \label{eq:l1}\\
    &=\log \left[ \det\Big((u\I - A)^{-1}\Big)^{m-j}\prod_{v\in\setb} \det\Big(\uhat \I - (A+vv^{*}) \Big) \right] \nonumber\\
    &=\log \left[ \det\Big((u\I - A)^{-1}\Big)^{m-j} \prod_{v\in\setb} \det(\uhat\I - A)(1-v^{*}(\uhat\I-A)^{-1}v) \right] \label{eq:l2}\\
    &=\log \left[ \det\Big((u\I - A)^{-1}\Big)^{m-j}\det(\uhat\I-A)^{m-j}\prod_{v\in\setb} (1-v^{*}(\uhat\I-A)^{-1}v) \right] \nonumber\\
    &=\log \left[ \det\Big(\Big((u\I - A)^{-1}(\uhat\I-A)\Big)^{m-j}\Big)\prod_{v\in\setb} (1-v^{*}(\uhat\I-A)^{-1}v) \right] \label{eq:l3}\\
    &=\log \left[\det\Big(\Big(u\I - A)^{-1}(\uhat\I-A)\Big)^{m-j}\Big)\right] + \log\left[\prod_{v\in \setb} (1-v^{*}(\uhat\I-A)^{-1}v)\right]\nonumber\\
    &=\trace\left[ \log\Big((u\I - A)^{-1}(\uhat\I-A)\Big)^{m-j}\right] + \sum_{v \in \setb} \log\Big( 1-v^{*}(\uhat\I-A)^{-1}v\Big)\nonumber\\
    &=(m-j)\cdot \trace \left[\log\Big((u\I - A)^{-1}(\uhat\I-A)\Big)\right] + \sum_{v\in \setb} \log\Big( 1-v^{*}(\uhat\I-A)^{-1}v\Big),\nonumber
\end{align}
where \eqref{eq:l1} and \eqref{eq:l3} follow by Lemma~\ref{lem:det_product}, and \eqref{eq:l2} follows by Matrix Determinant Lemma~(Lemma~\ref{lem:MDL}). This 
 completes the proof.
\end{proof}

We examine the second term in \eqref{eq:lowerbd1}, which involves the set of all the unpicked vectors $v\in\setb$. The following lemma derives a ``universal'' lower bound of this term, in the sense that this bound is only a function of $A$ instead of individual vectors in $\setb$.

\begin{lemma}
Assuming Lemma~\ref{lem:gap_bound} holds up to iteration $j$, we have  for the constructed  $A=A_j$  that
 \[
      \sum_{v\in\setb} \log\Big(1-v^{*}(\uhat\I-A)^{-1}v\Big) \geq \frac 1 {\alpha}\cdot  \trace\left[(\I - A) \cdot \log\Big(\I-\alpha(\uhat\I-A)^{-1}\Big)\right].
 \]
\label{lem:aks_bound}
\end{lemma}
\begin{proof}
    By Lemma~\ref{lem:gap_bound} we have that
    \[
    v^*(\uhat \I-A)^{-1}v \le \alpha \lambda_{\max}((\uhat \I-A)^{-1})=\frac{\alpha}{\lambda_{\min}(\uhat \I - A)} = \frac {\alpha}{\delta_u + (u-\lambda_{\max}(A))}\le \frac {\alpha}{\delta_u+1/3} \le 3\alpha<1.
    \]
    Hence, we can  apply the
    scalar and 
    matrix Maclaurin series of $\log\left(1-x\right)$  and have that
\begin{align}
    \sum_{v\in\setb} \log\Big(1-v^{*}(\uhat\I-A)^{-1}v\Big) &= -\sum_{v\in \setb}\sum_{k=1}^\infty \frac{\Big(v^{*}(\uhat\I-A)^{-1}v\Big)^k}{k} \nonumber\\
    &= - \sum_{v\in \setb}\sum_{k=1}^\infty \frac{\trace\Big[\Big((\uhat\I-A)^{-1}vv^{*}\Big)^k\Big]}{k} \nonumber \\
    &\ge- \sum_{v\in \setb}\sum_{k=1}^\infty \frac{\trace\Big[(\uhat\I-A)^{-k}(vv^{*})^k\Big]}{k} \label{eq:trace_ineq} \\
    &= -\sum_{k=1}^\infty  \sum_{v\in \setb}\frac{\trace\Big[(\uhat\I-A)^{-k}(vv^{*})^k\Big]}{k} \nonumber \\
    &= -\sum_{k=1}^\infty  \sum_{v\in \setb}\frac{\trace\Big[\alpha^{k-1}(\uhat\I-A)^{-k}(vv^*)\Big]}{k} \label{eq:norm}\\
    &= -\sum_{k=1}^\infty  \frac{\trace\Big[\alpha^{k-1}(\uhat\I-A)^{-k}\sum_{v\in \setb}(vv^*)\Big]}{k} \nonumber \\
    &= -\sum_{k=1}^\infty  \frac{\trace\Big[\alpha^{k-1}(\uhat\I-A)^{-k} (\I-A)\Big]}{k} \nonumber \\
    &= -\sum_{k=1}^\infty  \frac{\trace\Big[\frac 1 {\alpha}\Big(\frac {1}{\alpha}(\uhat\I-A)\Big)^{-k} (\I-A)\Big]}{k} \nonumber\\
    &= \frac 1 {\alpha}\cdot \trace\left[-\sum_{k=1}^\infty  \frac{\Big(\alpha(\uhat\I-A)^{-1}\Big)^{k}}{k}\cdot (\I-A)\right] \nonumber \\
    &= \frac 1 {\alpha} \trace \left[(\I-A)\cdot \log\Big(\I-\alpha(\uhat\I-A)^{-1}\Big)\right]\nonumber,
\end{align}
where \eqref{eq:trace_ineq} holds by the fact that $\trace((XY)^k)\leq \trace(X^k Y^k)$ holds for PSD matrices $X,Y$ and $k\in\mathbb{Z}^+$, and 
\eqref{eq:norm} holds by the fact that $\|v\|^2= \alpha$ for every $v\in B$. 
\end{proof}

Next we  apply  Chebyshev's sum inequality to replace the factor of $(\I-A)$ in the lower bound given in Lemma~\ref{lem:aks_bound} with the scalar $(m-j)/m$.

\begin{lemma}[Chebyshev's Sum Inequality] \label{lem:sum_prod_decreasing_funcs}
    Let $f(x): \R \xrightarrow{} \R$ and $g(x): \R \xrightarrow[]{} \R$ be both monotone increasing~(or both monotone decreasing) functions on $x$. Then, it holds for  $S \subset \R$ that 
    \[
        \sum_{x \in S} g(x) f(x) \geq \frac{1}{\cardinality{S}}\left( \sum_{x \in S} g(x)\right) \left(\sum_{x \in S} f(x)\right).
    \]
\end{lemma}

\begin{lemma} \label{lem:i_minus_a_term}
Assuming Lemma~\ref{lem:gap_bound} holds up to iteration $j$, we have  for the constructed $A=A_j$ that
    \begin{align*}
\lefteqn{\sum_{v\in \setb}  \Phi^u(A) - \Phi^{\uhat} (A+vv^{*})}\\
     & \geq (m-j)\cdot \trace\left[ \log\Big((u\I - A)^{-1}(\uhat\I-A)\Big)\right] + \frac 1 {\alpha}\frac{m - j}{m} \cdot\trace \left[ \log \left(\I - \alpha(\uhat\I - A)^{-1}\right)\right].
    \end{align*}
\end{lemma}
\begin{proof}
By Lemma~\ref{lem:pot_drop_part1} and Lemma~\ref{lem:aks_bound}, it is sufficient to show that
        \[
     \trace \left[(\I-A)\cdot \log\Big(\I-\alpha(\uhat\I-A)^{-1}\Big)\right] \geq \frac{m - j}{m} \cdot\trace \left[ \log \left(\I - \alpha(\uhat\I - A)^{-1}\right)\right].
    \]
    Similar to the proof of Lemma~\ref{lem:aks_bound}, we have $\norm{\alpha(\uhat \I-A)^{-1}}<1$ and
    \begin{align*}
     \trace \left[(\I-A)\cdot \log\Big(\I-\alpha(\uhat\I-A)^{-1}\Big)\right] = \sum_{i = 1}^d \left(1 - \lambda_i(A)\right)\cdot \log\left(1 - \frac{\alpha}{\uhat - \lambda_i(A)}\right).
    \end{align*}
    Applying Lemma~\ref{lem:sum_prod_decreasing_funcs} with $g(x) = (1 - x)$ and $f(x) = \log\left(1 - \alpha (\uhat - x)^{-1}\right)$ gives
    \begin{align*}
     \trace \left[(\I-A)\cdot \log\Big(\I-\alpha(\uhat\I-A)^{-1}\Big)\right] & \geq \frac{1}{d} \left( \sum_{i = 1}^d (1 - \lambda_i(A) ) \right) \left(\sum_{i = 1}^d \log\left(1 - \frac{\alpha}{\uhat - \lambda_i(A) }\right)\right) \\
     & = \frac{1}{d} \trace \left(\I - A\right) \trace\left[ \log \left(\I - \alpha(\uhat \I - A)^{-1}\right)\right].
    \end{align*}
    In iteration $j$, we have
    \begin{align*}
        \frac{1}{d} \cdot \trace \left(\I - A\right) & = \frac{1}{d} \cdot\left(d - j \cdot \alpha\right) \\
        & = \frac{1}{d}\cdot \left((m - j)\cdot  \alpha\right) \\
        & = \frac{m-j}{m},
    \end{align*}
    where we use the fact that $\trace(A_j) = j \alpha$ and $m\alpha =d$. Combining the two inequalities above proves the statement.
\end{proof}

Now we are ready to apply 
Lemma~\ref{lem:condition} and 
Lemma~\ref{lem:i_minus_a_term} to prove Lemma~\ref{lem:gap_bound}. 
For this proof, we explicitly use index $j$ for all the matrices, sets, and variables used in iteration $0
\leq j<m/2$.

\begin{proof}[Proof of Lemma~\ref{lem:gap_bound}]
We proceed by induction.  
For the base case $j=0$, we have $u_0=1/2$ and $\lambda_i(A_0)=0$ for every $1\leq i\leq d$ and therefore the the statement holds.

For the inductive step, we  assume that the statement holds for all $j'\le j$ and show that $c_{j+1}\ge 1/3$ and $\kappa(\uhat_{j+1}-A_{j+1})\le 2/3$ as long as $j< m/2$.
We achieve this using Lemma~\ref{lem:condition}, which will first require us to lower bound $\det(u_{j+1}\I-A_{j+1})$. We show that we can arrive at a lower bound for $\det(u_{j+1}\I-A_{j+1})$ as a consequence of upper bounding the average potential increase.
We first upper bound the average potential increase in terms of $\alpha$ and $c_j$. From Lemma~\ref{lem:i_minus_a_term}, we have that
\begin{align*}
    \lefteqn{\frac{1}{\cardinality{\setb_j}} \sum_{v \in \setb_j} \Phi^{u_j}(A_j) - \Phi^{u_{j+1}}(A_j + v v^*)} \\
    & \geq \trace \left[\log\left((u_j\I - A_j)^{-1}(u_{j+1}\I-A_j)\right) \right]+ \frac{1}{\alpha m}\cdot  \trace \left[\log\left(\I-\alpha(u_{j+1}\I-A_j)^{-1}\right)\right] \\
    & = \trace\left[ \log\left((u_j\I - A_j)^{-1}(u_{j+1}\I-A_j)\right)\right] + \frac{1}{d}\cdot \trace\left[ \log\left(\I-\alpha(u_{j+1}\I-A_j)^{-1}\right)\right] \\
    & = \frac{1}{d} \left[d \cdot \trace\left[ \log\left((u_j\I - A_j)^{-1}(u_{j+1}\I-A_j)\right)\right] + \trace\left[ \log\left(\I-\alpha(u_{j+1}\I-A_j)^{-1}\right) \right]\right].
\end{align*}
By writing $(u_j - x)^{-1} (u_j + \delta_u - x) = 1 + \delta_u (u_j - x)^{-1}$ we have that
\[
    \trace\left[\log\left((u_j\I - A_j)^{-1}(u_{j+1}\I - A_j)\right)\right] = \trace\left[\log\left(\I + \delta_u(u_j\I - A_j)^{-1}\right)\right],
\]
and therefore 
\begin{align*}
    \lefteqn{\frac{1}{\cardinality{\setb_j}} \sum_{v \in \setb_j} \Phi^{u_j}(A_j) - \Phi^{u_{j+1}}(A_j + v v^*)} \\
    & \geq \frac{1}{d} \left[d \cdot \trace\left[ \log\left(\I + \delta_u (u_j\I - A_j)^{-1}\right)\right] + \trace \left[\log\left(\I-\alpha(u_{j+1}\I-A_j)^{-1}\right)\right] \right] \\
    & = \frac{1}{d} \left[\trace\left[ \log\left(\I + \delta_u (u_j\I - A_j)^{-1}\right)^d\right] + \trace \left[\log\left(\I-\alpha(u_{j+1}\I-A_j)^{-1}\right) \right]\right] \\
    & \geq \frac{1}{d} \left[\trace\left[ \log\left(\I + \alpha(u_j\I - A_j)^{-1}\right)\right] + \trace\left[ \log\left(\I-\alpha(u_{j+1}\I-A_j)^{-1}\right) \right]\right] \\
    & = \frac{1}{d} \sum_{i = 1}^d \left[ \log\left(1 + \frac{\alpha}{u_j - \lambda_i(A_j)}\right) + \log\left(1 - \frac{\alpha}{u_j + \delta_u - \lambda_i(A_j)}\right)\right] \\
    & \geq \frac{1}{d} \sum_{i = 1}^d \left[ \log\left(1 + \frac{\alpha}{u_j - \lambda_i(A_j)}\right) + \log\left(1 - \frac{\alpha}{u_j - \lambda_i(A_j)}\right)\right] \\
    & = \frac{1}{d} \sum_{i = 1}^d \log\left(1 - \frac{\alpha^2}{(u_j - \lambda_i(A_j) )^2} \right).
\end{align*}
Since $\log(1 - \alpha^2 x^{-2})$ is increasing on $x$, we have that 
\begin{align*}
\frac{1}{\cardinality{\setb_j}} \sum_{v \in \setb_j} \Phi^{u_j}(A_j) - \Phi^{u_{j+1}}(A_j + v v^*)
    & \geq \frac{1}{d} \sum_{i = 1}^d \log\left(1 - \frac{\alpha^2}{(u_j - \lambda_i(A_j))^2} \right) \\
    & \geq \frac{1}{d} \sum_{i = 1}^d \log\left(1 - \frac{\alpha^2}{c_j^2} \right) \\
    & = \log \left(1 - \alpha^2 c_j^{-2}\right) \\
    & \geq - \frac{\alpha^2 c_j^{-2}}{1 - \alpha^2 c_j^{-2}} \\
    & \geq - \frac{2 \alpha^2}{c_j^2},
\end{align*}
where we use the assumption that $c_j \geq 1/3>\sqrt{2} \alpha$.  Since the vector  picked by Algorithm~\ref{alg:alg1}   increases the potential value no more than the average, we have that
\begin{align*}
    - \log \det (u_j \I - A_j) + \log \det (u_{j+1} \I - A_{j+1}) \geq - \frac{2 \alpha^2}{c_j^2},
\end{align*}
which is equivalent to
\[
    \det(u_{j+1}\I - A_{j+1}) \geq \exp\left(- \frac{2 \alpha^2}{c_j^2}\right) \cdot \det(u_j \I - A_j).
\]
By combining these inequalities for every iteration up to $j$, we have
\begin{equation} \label{eq:det}
    \det(u_{j+1}\I - A_{j+1}) \geq \exp\left(- \sum_{j' = 0}^{j} \frac{2 \alpha^2}{c_{j'}^2}\right) \cdot \det(u_0 \I).
\end{equation}
Now we apply the condition number inequality in Lemma~\ref{lem:condition}. Let
    \[
        x_{j+1} = \sqrt{1 - \left(\frac{d}{\trace(u_{j+1} \I - A_{j+1})}\right)^d \cdot \det(u_{j+1} \I - A_{j+1})}.
    \]
    From \eqref{eq:det}, we have
    \[
        x_{j+1} \leq \sqrt{1 - \exp\left(- \sum_{j' = 0}^j \frac{2 \alpha^2}{c_{j'}^2} \right) \left(\frac{d}{\trace(u_{j+1} \I - A_{j+1})}\right)^d \cdot \det(u_0 \I)}.
    \]
    Furthermore, notice that
    \begin{align*}
        \trace(u_{j+1} \I - A_{j+1}) & = d\cdot  (u_0 + (j+1) \cdot\delta_u) - (j+1) \cdot\alpha \\
        & = \frac{d}{2} + (j+1)\cdot \left(d \delta_u - \alpha\right) \\
        & = \frac{d}{2},
    \end{align*}
    where the final equality follows by the definition of $\delta_u = \alpha / d$.
    We also have that
    \[
        \det(u_0 \I) = \frac{1}{2^d}.
    \]
    Combined with the inductive hypothesis, we have
    \begin{align*}
        x_{j+1} & \leq \sqrt{1 - \exp\left(- \sum_{j' = 0}^j \frac{2 \alpha^2}{c_{j'}^2} \right) } \\
        & \leq \sqrt{1 - \exp\left(- 18 \sum_{j' = 0}^j \alpha^2\right)} \\
        & \leq \sqrt{1 - \exp\left(- 9 m \alpha^2\right)} \\
        & = \sqrt{1 - \exp\left(- 9 d \alpha \right)} \\
        & \leq \sqrt{1 - \exp\left(- \frac{9}{221} \right)}\\
        & \leq \frac{1}{5}.
    \end{align*}
    Then, it holds by Lemma~\ref{lem:condition} that 
    \[
        \kappa(u_{j+1} \I - A_{j+1}) \leq \frac{1 + x_{j+1}}{1 - x_{j+1}} \leq \frac{3}{2}.
    \]
    We also have that
    \[
        \lambda_{\mathrm{max}}(u_{j + 1}\I - A_{j+1}) \geq \frac{1}{d} \cdot \trace(u_{j+1}\I - A_{j+1}) = \frac{1}{2}.
    \]
    Then,
    \[
        c_{j+1} = \lambda_{\mathrm{min}}(u_{j+1}\I - A_{j+1}) \geq \frac{\lambda_{\mathrm{max}}(u_{j+1} \I - A_{j+1})}{\kappa(u_{j+1}\I - A_{j+1})} \geq \frac{1}{3},
    \]
    which completes the inductive step.
\end{proof}

We believe the tools used to prove Theorem~\ref{theorem:alg1} are interesting in their own right. It is known that the only non-trivial matrix function with a simple formula for rank-1 updates is $f(x)=x^{-1}$ \cite{Beckermann2017LowrankUO}. Previous potential function-based sparsification algorithms exploit this formula and the potential function $\Psi^u(A)=\trace (u\I-A)^{-1}$ in their analysis~\cite{BSS}. Sparsification can be seen as a relaxation of the $\mathsf{KS}_2$ problem in the sense that sparsification algorithms are free to add different scalar multiples of some $vv^*$ in each iteration and use the same vector multiple times. The  analyses of these algorithms appear to become intractable when these relaxations are removed, partly due to the use of the Sherman-Morrison formula for the inverse of rank-1 updates. We overcome this apparent intractability using the logarithmic potential function $\Phi^u(A)=\trace \log(u\I-A)^{-1}$, the Matrix Determinant Lemma and Lemma~\ref{lem:aks_bound}.
In addition, instead of examining all the vectors of $\mathcal{I}$ in each iteration, we analyse the change of potential functions over all the unpicked vectors in each iteration, ensuring that the set of picked vectors and the set of unpicked ones always forms a partition; this is another key difference between our analysis and the related ones in the  literature.

\section{The Second Algorithm and Analysis\label{sec:second}} 

In this section, we present Algorithm~\ref{alg:algo2} which also achieves the theoretical guarantees stated in Theorem~\ref{theorem:alg1}.
We can view Algorithm~\ref{alg:algo2} as a modification of Algorithm~\ref{alg:alg1} in which we fix
the barrier value  to be $1$, and set $\delta_u$ to be $0$.
Thus, in each iteration we select the vector $v$ which maximizes the determinant $\det(\I - A_j - vv^*)$.
By the Matrix Determinant Lemma (Lemma~\ref{lem:MDL}), this is equivalent to choosing $v$ to minimize the quadratic form $v^*(\I - A_j)^{-1}v$ in each iteration.  
\begin{algorithm}[th] 
\SetAlgoCaptionSeparator{}
\caption{\label{alg:algo2}}
\textbf{Input}: $\mathcal{I} = \{v_i\}_{i = 1}^m$, where $v_i \in \C^d$ and $\norm{v_i}^2 = \alpha$ \label{alg:nobarrier}\\
$A_0 \gets \mathbf{0}_{d \times d}$, $\seta_0\gets \emptyset$ \\
$\setb_0 \gets \mathcal{I}$ \\
\For{$j = 0$  to  $m/2-1$}{
    $v_j \gets \argmin_{v \in \setb_{j}} v^* (\identity - A_j)^{-1} v$ \\
    $A_{j+1} \gets A_{j} + v_j v_j^*$ \\
    $\seta_{j+1} \gets \seta_{j} \cup \{v_j\}$\\
    $\setb_{j+1} \gets \setb_{j} \setminus \{v_j\}$\\
}
\Return $\seta_{m/2}$
\end{algorithm}

Rather than considering the change in the potential function $\Phi^{1}(A_j)$ in each iteration, in Lemma~\ref{lem:nobarrierdet} we directly bound the value of
\[
    \Phi^1(A_{m/2}) = \trace \log(\I - A_{m/2})^{-1} = -\log \det(\I - A_{m/2}),
\]
where $A_{m/2}$ is the matrix corresponding to the output of   Algorithm~\ref{alg:nobarrier}.

\begin{lemma} \label{lem:nobarrierdet}
Let $A_{m/2}$ be the matrix corresponding to $\seta_{m/2}$. Then, it holds that
    \[
        \log \det(\identity - A_{m/2}) \geq - d \log(2) -\frac{2 d^2}{m}.
    \]
\end{lemma}

\begin{proof}
    In each iteration of Algorithm~\ref{alg:algo2}, we have that
    \[
        \sum_{v \in \setb_{j}} v^* (\identity - A_j)^{-1} v = \trace \left((\identity - A_j)^{-1} (\identity - A_j)\right) = \trace(\identity) = d.
    \]
    Therefore,  it holds for the selected $v_j$ that
    \begin{equation} \label{eq:detupdate}
        v_j^* (\identity - A_j)^{-1} v_j \leq \frac{1}{\cardinality{B_{j}}} \sum_{v \in \setb_{j}} v^* (\identity - A_j)^{-1} v = \frac{d}{m-j}.
    \end{equation}
    By the Matrix Determinant Lemma~(Lemma~\ref{lem:MDL}), we have that
    \begin{align*}
    \det(\identity - A_{j+1}) &= (1 - v_j^*(\identity - A_j)^{-1} v_j) \det(\identity -A_j) \\
        & \geq \left(1 - \frac{d}{m-j}\right) \det(\identity - A_j) \\
        & \geq \left[\prod_{i=0}^{j} \left(1 - \frac{d}{m-i}\right)\right] \det(\identity - A_0) \\
        & = \prod_{i=0}^{j}\left(1 - \frac{d}{m-i}\right),
    \end{align*}
    where the second inequality follows by applying \eqref{eq:detupdate} a further $j$ times.
    Then, we have
    \begin{align*}
        \log \det(\I - A_{m/2}) & \geq \sum_{i = 0}^{m/2-1} \log\left(1 - \frac{d}{m-i}\right) \\
        & \geq \int_{i=0}^{\frac{m}{2}} \log\left(1 - \frac{d}{m-i}\right) \mathrm{d}i \\
        & = \int_{i=0}^{\frac{m}{2}} \log(m - i - d) \mathrm{d}i - \int_{i=0}^{\frac{m}{2}} \log(m-i) \mathrm{d}i \\
        & = \big[- (m-i-d) \log(m-i-d) + (m-i-d) + (m-i) \log(m-i) - (m-i) \big]_{i=0}^{m/2}\\
        & = \big[(m-i-d+d) \log(m-i) -(m-i-d)\log(m-i-d) - d\big]_{i=0}^{m/2}\\
        & = \left[d \log(m - i) - (m - i - d) \log\left(1 - \frac{d}{m-i}\right) - d \right]_{i=0}^{m/2} \\
        & = d \log\left(\frac{m}{2}\right) - \left(\frac{m}{2} - d\right) \log\left( 1 - \frac{2d}{m} \right) - d \log(m) + (m-d) \log\left(1 - \frac{d}{m}\right)\\
        & \geq - d \log(2) + \left(\frac{m}{2}-d\right) \frac{2d}{m} - \frac{(m-d) (d/m)}{1-(d/m)}\\
        & = -d \log(2) + d - \frac{2d^2}{m} - d \\
        & = - d \log(2) - \frac{2d^2}{m},
    \end{align*}
       where the second line follows since $\log(1- d/(m-i))$ is a decreasing function on $i$,
    the second equality follows since $\int \log(x - y) \mathrm{d}y = -(x-y)\log(x-y) + (x-y) + C$,
    and for the final inequality we use the fact that $-x/(1-x) \leq \log(1-x) \leq -x$ for all $0 \leq x < 1$.
    This completes the proof.
\end{proof}

With this bound on the determinant of $\left(\I - A_{m/2}\right)$, we are able to apply Lemma~\ref{lem:condition} to show that Algorithm~\ref{alg:algo2} also achieves the guarantee given in Theorem~\ref{theorem:alg1}.
\begin{proof}[Proof of Theorem~\ref{theorem:alg1}.]
    Since $m \alpha = d$, we have that 
    \[
        \trace(\identity - A_{m/2}) = d - \frac{m \alpha}{2} = \frac{d}{2}.
    \]
    Then, by Lemma~\ref{lem:nobarrierdet} and the condition of $m \geq 49 d^2$, we have that
    \begin{align*}\log\left[\left(\frac{d}{\trace(\identity - A_{m/2})}\right)^d \det(\identity - A_{m/2}) \right] & \geq d \log(2)  - d \log(2) - \frac{2 d^2}{m} = - \frac{2d^2}{m} \geq \log\left(\frac{24}{25}\right).
    \end{align*}
    Then, let
    \[
        x \triangleq \sqrt{1 - \left(\frac{d}{\trace(\identity - A_{m/2})}\right)^d \det(\identity - A_{m/2})} \leq \sqrt{1 - \frac{24}{25}} = \frac{1}{5}.
    \]
    By Lemma~\ref{lem:condition}, we can bound the condition number of $\identity - A_{m/2}$ as
    \[
        \kappa(\identity - A_{m/2}) \leq \frac{1 + x}{1 - x} \leq \frac{3}{2}.
    \]
    Then, we have that 
    \[
        \lambda_{\mathrm{max}}(\identity - A_{m/2}) \leq \kappa(\identity - A_{m/2}) \cdot \lambda_{\mathrm{min}}(\identity - A_{m/2}) \leq \frac{3}{2} \cdot \frac{1}{2} = \frac{3}{4},
    \]
    and
    \[
        \lambda_{\mathrm{\min}}(\identity - A_{m/2}) \geq \frac{\lambda_{\mathrm{max}}(\identity - A_{m/2})}{\kappa(\identity - A_{m/2})} \geq \frac{2}{3} \cdot \frac{1}{2} = \frac{1}{3}.
    \]
    Since $\lambda_{\mathrm{max}}(A_{m/2}) = 1 - \lambda_{\mathrm{min}}(\identity - A_{m/2})$ and $\lambda_{\mathrm{min}}(A_{m/2}) = 1 - \lambda_{\mathrm{max}}(\identity - A_{m/2})$,
    we have that
    \[
        \frac{1}{4} \leq u^* A_{m/2} u \leq \frac{2}{3},
    \]
    for all unit vectors $u \in \C^d$. The time complexity analysis of the algorithm follows the one for Algorithm~\ref{alg:alg1}.
\end{proof}

Finally, as our presented algorithms  are based on maximizing the determinant of the constructed matrix $A$ (or $\I - A$) and applying the condition number inequality~(Lemma~\ref{lem:condition}),
we discuss the difference between our technique   with a direct application of more general determinant maximisation algorithms.
Recall that
in the cardinality-based determinant maximisation problem, we are given as input a set of vectors $\mathcal{I} = \{v_i\}_{i=1}^m$ and an integer $k$, and the objective is to find
\[
    \widehat{S} = \argmax_{\substack{S \subset \mathcal{I} \\ \cardinality{S} \leq k}} \det\left(\sum_{v \in S} v v^*\right).
\]
The state-of-the-art polynomial-time approximation algorithm is presented by Singh and Xie~\cite{SX-determinant} and achieves an $\exp(d)$-approximation of the optimal solution.

By the result of Marcus et al.\ \cite{MSS}, we know that for any set $\mathcal{I} = \{v_i\}_{i=1}^m$ such that $\sum_{v \in \mathcal{I}} v v^* = \I$ and $\norm{v_i} = \alpha$ for all $i \in [m]$, there exists a partition $S_1, S_2$ of $\mathcal{I}$ such that
\[
    \norm{\frac{1}{2}\cdot \I - \sum_{i \in S_j} v_i v_i^*} \leq 3 \sqrt{\alpha}
\]
for $j \in \{1, 2\}$. Without loss of generality, let $\cardinality{S_1} \geq \cardinality{S_2}$. Then, we have that
\[
    \det\left(\sum_{i \in S_1} v_i v_i^*\right) \geq \left(\frac{1}{2} - 3\sqrt{\alpha}\right)^{d/2} \left(\frac{1}{2} + 3 \sqrt{\alpha}\right)^{d/2} = 2^{-d} \exp\left(-\Theta\left(\frac{d^2}{m}\right)\right).
\]
Thus, the guarantee on the determinant given by Lemma~\ref{lem:nobarrierdet} matches the determinant guarantee implied by \cite{MSS}. Applying Lemma~\ref{lem:condition}, we achieve a non-trivial bound on the eigenvalues of the constructed matrix when $m = \Omega\left(d^2\right)$.
On the other hand, applying the algorithm  by Singh and Xie~\cite{SX-determinant} returns a set $S$ such that $\det\left(\sum_{i \in S} v_i v_i^*\right) \geq 2^{-d} \exp\left(-O\left(d^3/m\right)\right)$, and gives a non-trivial bound on the eigenvalues of the constructed matrix only when $m = \Omega\left(d^3\right)$.

\section*{Acknowledgement}
This work is supported by    EPSRC Early Career Fellowship~(EP/T00729X/1) and 
 EPSRC Doctoral Training Studentship (2590711).  
Part of this work was done when Peter Macgregor and He Sun were visiting the Simons Institute for the Theory of Computing in Fall 2023.

\bibliographystyle{alpha}
\bibliography{references.bib}

\end{document}